\documentclass[12pt,a4paper]{article}

\usepackage{latexsym}
\usepackage{graphicx}
\usepackage{amssymb}
\usepackage{amsmath}
\usepackage{amsfonts}
\usepackage{color}

\newtheorem{theorem}{Theorem}[section]
\newtheorem{lemma}{Lemma}[section]

\newtheorem{definition}{Definition}[section]

\newtheorem{assumption}{Assumption}[section]

\def\QED{\mbox{\rule[0pt]{1ex}{1ex}}}
\def\Q{\hspace*{\fill}~\QED\par\endtrivlist\unskip}

\title{A note on the stability of multiclass Markovian queueing networks}
\author{Sayee C. Kompalli and Ravi R. Mazumdar \\
Department of Electrical and Computer Engineering \\
University of Waterloo, Canada N2L 3G1. \\
Email: skompall@uwaterloo.ca, mazum@ece.uwaterloo.ca
}

\begin{document}
\maketitle

\begin{abstract}
In this paper we show that in a multiclass Markovian network with unit rate servers, the condition that the average load $\rho$ at every  server is less than unity  is indeed sufficient for the stability or positive recurrence for \emph{any} work conserving scheduling policy and \emph{class-independent} routing. We use a variation of the positive recurrence criterion for multidimensional discrete-time Markov chains over countable state spaces due to Rosberg (JAP, Vol.~17, No.~3, 1980) and a monotonicity argument to establish this assertion.
\end{abstract}

\section{Introduction and Overview} \label{sec:literature survey}
In this paper we show that the condition that the load $\rho = \sum_{i=1}^{J _m}\rho_i <1$  where $J_m$ denotes the number of classes that visit node $m$ at every work conserving unit rate server of a open multiclass Markovian network is sufficient to ensure the stability of the network with class-independent routing.
We prove this via a variation of the positive recurrence criterion for multidimensional discrete-time Markov chains over countable state spaces due to Rosberg \cite{Ros-JRN-JAP-1980}.

Much research  effort has focussed on studying the stability of multiclass queueing networks in the last two decades and the amount of the literature thus is understandably vast. Much of the research on this topic was fueled by the  example of a queueing network in a deterministic setting that Lu and Kumar studied  in  \cite{LuKum-JRN-ACTRAN}. This paper showed the existence of a scheduling policy under which a certain network becomes unstable even though the load on each work conserving server in the network is less than unity, which was counter intuitive given the well known fact that the load is less than unity is a sufficient and necessary condition for the stability of a single queue with multiple classes under \emph{any} work conserving policy. The appearence of the paper \cite{LuKum-JRN-ACTRAN} was followed by a number of papers \cite{KumSei-JRN-ACTRAN, KumMey-JRN-JSAC-1995, RybSto-JRN-PRI, Bra-AAP-1994, Sei-JRN-ACTRAN, Dai-AAP-1995, Sto-MPRF-1995, Bra-QUESTA-1996, FosRyb-MPRF-1996} that either established the positive assertion that multiclass networks are stable under specific settings and policies, or proved that under certain specific scheduling policies for which the usual traffic condition $\rho < 1$ at every work conserving server is \emph{not} sufficient to ensure the stability of the network.

Foss and Rybko \cite{FosRyb-MPRF-1996} proposed a condition, called \emph{Jackson-type} condition, under which they showed that $\rho < 1$ at every server is a sufficient condition when a FIFO scheduling policy used under renewal assumptions on arrival and service processes. The Jackson-type condition roughly means that the routing mechanism is ``almost independent'' of customer classes. In this paper we also assume that the \emph{routing is class independent} and we show that  when arrivals are Poisson and service times are exponential, the network is stable under the usual traffic conditions for \emph{any} work conserving scheduling policy (that could differentiate between classes). In the network models considered in \cite{LuKum-JRN-ACTRAN}, \cite{Sei-JRN-ACTRAN}, and \cite{Bra-AAP-1994}, it should be noted  that the customers follow a predetermined fixed route before exiting the network and hence the routing is not class independent.This is different from the situation considered by Kelly \cite{Kelly-BOOK} where it is shown that multiclass Markovian networks with \textit{class independent} routing are stable under FIFO and PS scheduling under the {\em usual} condition. This paper essentially extends those results to any work-conserving scheduling policy.


\section{The Stability Analysis} \label{sec:stability analysis}
A common technique to establish ergodicity of countable space Markov chains is through Foster-Lyapunov approach which consists of finding an appropriate Lyapunov function that satisfies Foster's criterion~\cite{Foster_1953}.  Though this approach has proved
to be very successful for one dimensional Markov chains, finding such a function proved to be very difficult in the case of multidimensional Markov chains especially networks. Rosberg~\cite{Ros-JRN-JAP-1980} extended Foster's criterion~\cite{Foster_1953} for $J\geq2$-dimensional Markov chains  by requiring existence of $J$ Lyapunov functions, one for each coordinate of the process. But the applicability of his criterion becomes limited when the state definition for the Markov chain does not have a \textit{fixed} dimension. For example, the well known FCFS scheduling policy requires a variable-length dimension for the state definition. Also, depending on the modeling assumptions, one may only require a fewer component Lyapunov functions than the dimensionality of the underlying state space.  Thus, we propose a variation of Rosberg's criterion that expands the scope of the applicability that we use to establish stability for the multiclass network model we consider.

Let $\mathcal{X}$ be a countable set of states over which the irreducible, aperiodic, and discrete-time Markov chain $\{X^{n}, n \geq 0\}$ takes its values. For any integer $k \geq 1$, define $\left\{p_{xy}^{k}, x, y \in
\mathcal{X} \right\}$ to be the $k$-step transition probability law of the
Markov chain $\{X^{n}, n \geq 0\}$. For any subset $\mathcal{B} \subseteq
\mathcal{X}$, we know that the $\lim_{k \rightarrow \infty}
p_{x\mathcal{B}}^{k} = \lim_{k \rightarrow \infty} \sum_{y \in \mathcal{B}}
p_{xy}^{k}  = \pi(\mathcal{B}) \geq 0$ exists and is independent of the initial
state $x$. For any nonnegative-valued function $V$ on $\mathcal{X}$, let us
define $\Delta^{k} V(x) \triangleq \sum_{y} p_{xy}^{k}V(y) - V(x)$ to be the
$k$-step drift of the function $V$ in state $x$. Let $c^{*} \geq 0$ denote the limit $\lim_{k \rightarrow \infty} \frac{\Delta^{k}V(x)}{k}$, and is independent of $x$. The following definition was originally proposed in \cite{Ros-JRN-JAP-1980}.

\begin{definition}
\label{def:uub}
The sequence $\left\{ \frac{\Delta^{k}V(x)}{k}, k \geq 1\right\}$ is said to be uniformly upper bounded (UUB) if, for any $\delta > 0$, there exists a positive integer $\mathsf{K}$ such that $\frac{\Delta^{k}V(x)}{k} < c^{*} + \delta$ for $k \geq \mathsf{K}$ and $x \in \mathcal{X}$.
\end{definition}

We now propose a variation of the Rosberg's positive criterion \cite{Ros-JRN-JAP-1980}. In the rest of the paper, we will use the same notation $\mathcal{X}$ to denote both the state space and its subsets. The distinction is made through the usage of subscripts, i.e., $\mathcal{X}_{j}$ denotes a subset.

\begin{theorem}[A variation of Rosberg's criterion ~\cite{Ros-JRN-JAP-1980}]
\label{th:My generalization of Rosberg's theorem}
Let $J \geq 2$ be an integer.
\begin{assumption}
\label{assumption 1}
There exists a collection $\mathcal{P} = \left\{\mathcal{P}_{1}, \mathcal{P}_{2}, \ldots, \mathcal{P}_{J}\right\}$  of partitions of the set $\mathcal{X}$  where $\mathcal{P}_{j} = \left\{\mathcal{X}_{j}, \mathcal{X}_{j}^{c}\right\}$, and
nonnegative-valued functions $\{V_{j}(x), x \in \mathcal{X}\}$ for $1 \leq j \leq J$ such that the drift $\Delta V_{j}(x)$ of the function $V_{j}$ in the state $x$ has the following form:
\begin{eqnarray}
\label{one-step drift definition}
\Delta V_{j}(x) &\leq& \left\{
\begin{array}{ll}
\eta_{j} & \mbox{for}\;x \in \mathcal{X} \\
-\epsilon_{j} & \mbox{for}\;x \in \mathcal{X}_{j}^{c},
\end{array}
\right.
\end{eqnarray}
where $\epsilon_{j} > 0$ and $\eta_{j} \geq 0$.
\end{assumption}

\begin{assumption}
\label{assumption 2}
There exist partitions $\left\{\mathcal{A}_{j, k}, \mathcal{A}_{j, k}^{c}\right\}$, $k \geq 1$
and $1 \leq j \leq J$, of the set $\mathcal{X}$ with the following two properties:
\begin{itemize}
\item[(i)]   $p_{xy}^{l} = 0$, $0 \leq l \leq k-1$, for $x \in \mathcal{A}_{j, k}^{c}$ and $y \in \mathcal{X}_{j}$
\item[(ii)]  $\cap_{j} \mathcal{A}_{j, k}$ is a finite set
\end{itemize}
\end{assumption}

\begin{assumption}
\label{assumption 3}
For $1 \leq j \leq J$, the sequence $\left\{ \frac{\Delta^{k}V_{j}(x)}{k}, k \geq 1\right\}$
is uniformly upper bounded (UUB).
\end{assumption}

Then the Markov chain $\{X_{n}, n \geq 1\}$ is positive recurrent. \Q
\end{theorem}

Before we provide a proof of Theorem \ref{th:My generalization of Rosberg's theorem} we first discuss how our variation is different from the original theorem in \cite{Ros-JRN-JAP-1980}.

In \cite{Ros-JRN-JAP-1980} it is assumed that the the state space of the MC  model denoted by $\mathcal{X} =
\mathbb{Z}_{+}^{J}$, the $J$-dimensional non-negative integer space, for some integer $J \geq 2$ and an
\emph{equal number}  $J$  partitions denoted by $\left\{\mathcal{X}_{j} ,
\mathcal{X}_{j}^{c}\right\}$, $1 \leq j \leq J$, of the countable space
$\mathcal{X}$, and also the same number $J$ of Lyapunov functions
$\left\{V_{j}(x), x \in \mathcal{X}\right\}$, $1 \leq j \leq J$.  In the variation proposed in this paper, \textit{we do not require the countable space
$\mathcal{X}$ to have a fixed predetermined dimension}.  This corresponds to the
Assumption~\ref{assumption 1} of Theorem~\ref{th:My generalization of Rosberg's
theorem}.  Hence we are free to choose an appropriate number of Lyapunov
functions and the corresponding suitable partitions of the state space
$\mathcal{X}$.  We believe this generalization is useful  in many situations of interest when the state space is not of fixed predetermined dimension.

For instance, when the FIFO scheduling policy is implemented at a server  the state of each queue is characterized not just by the total number of customers present in the queue, but also by the class of the customers and the relative order of arrival of the customers of different classes. Thus the state description of the server is \emph{not} of fixed dimension but of variable dimension, and potentially can grow very large depending on the number of customers.



\begin{proof}
Since the proof of Theorem \ref{th:My generalization of Rosberg's theorem} is almost identical to the Proof of Theorem 2 in \cite{Ros-JRN-JAP-1980} we provide here only the main argument, omitting many intermediate results and supporting arguments which can be found in \cite{Ros-JRN-JAP-1980}.

Let us fix an arbitrary $\delta > 0$. From Assumption \ref{assumption 3} of Theorem \ref{th:My generalization of Rosberg's theorem}, it follows
that there exits a positive integer $\mathsf{K}$ such that $\frac{\Delta^{k}
V_{j}(x)}{k} \leq c_{j}^{*} + \delta$ for $k \geq \mathsf{K}$ and $1 \leq j
\leq J$. Let us pick one such $\mathsf{K}$, and then introduce the set of
functions $\left\{g_{j}^{\mathsf{K}}(x); x \in \mathcal{X}\right\}$ such that
the following holds:
\begin{eqnarray}
\label{eq:error function}
\Delta^{\mathsf{K}} V_{j}(x) &=& -g_{j}^{\mathsf{K}}(x) + \mathsf{K} \left(c_{j}^{*} + \delta \right)
\end{eqnarray}
Two observations on the functions $g_{j}^{\mathsf{K}}$ are in order: the first
and the obvious observation is that $g_{j}^{\mathsf{K}}(x) \geq 0$ for $x \in
\mathcal{X}$. Also, since $\Delta^{\mathsf{K}} V_{j}(x) \leq -\mathsf{K}
\epsilon_{j}$ for $x \in \mathcal{A}_{j,\mathsf{K}}^{c}$ (Lemma 2 of \cite{Ros-JRN-JAP-1980}), we have that
$g_{j}^{\mathsf{K}}(x) \geq \left(c_{j}^{*} + \delta + \epsilon_{j}\right)$ for
$x \in \mathcal{A}_{j,\mathsf{K}}^{c}$. Set $\epsilon = \min_{j} \epsilon_{j}$
and $\delta  = \min_{j} \delta_{j}$.  As a result, we have the obvious
deduction that $\max_{j} g_{j}^{\mathsf{K}}(x) \geq \min_{j} \mathsf{K} \left(
c^{*}_{j} + \delta + \epsilon_{j} \right) = \left( c^{*} + \delta + \epsilon
\right)$ for $x \in \cup_{j} \mathcal{A}_{j,\mathsf{K}}^{c}$. Hence $\max_{j}
g_{j}^{\mathsf{K}}(x) < \min_{j} \mathsf{K} \left( c^{*} + \delta + \epsilon
\right)$ implies that $x \in \cap_{j} \mathcal{A}_{j,\mathsf{K}}$. We should
note that $x \in \cap_{j} \mathcal{A}_{j,k}$ \emph{need not imply that}
$\max_{j} g_{j}^{\mathsf{K}}(x) < \mathsf{K} \left( c^{*} + \delta + \epsilon
\right)$.

Denote by $\mathsf{E}_{x}\left( g_{j}^{\mathsf{K}} (X^{n}) \right)$ the
expectation of $g_{j}^{\mathsf{K}} (X^{n})$ given that $X^{0} = x$ and by
$p_{x}(X^{n} \in A)$ the probability that $X^{n} \in A$ given that $X^{0} = x$.
Now

\begin{eqnarray*}
\frac{\Delta^{n\mathsf{K}} V_{j}(x)}{n} &=& \sum_{y \in \mathcal{X}} \frac{1}{n} \sum_{l=0}^{n-1} p_{xy}^{l\mathsf{K}} \Delta^{\mathsf{K}} V_{j}(y) \\
&\stackrel{(a)}{=}& \sum_{y \in \mathcal{X}} \frac{1}{n} \sum_{l=0}^{n-1} p_{xy}^{l\mathsf{K}} \left[-g^{\mathsf{K}}_{j}(y) + \mathsf{K}(c^{*}_{j} + \delta) \right] \\
&=& -\sum_{y \in \mathcal{X}} \frac{1}{n} \sum_{l=0}^{n-1} p_{xy}^{l\mathsf{K}} g^{\mathsf{K}}_{j}(y) + \mathsf{K}(c^{*}_{j} + \delta) \\
&=& -\frac{1}{n} \sum_{l=0}^{n-1} \mathsf{E}_{x} \left( g_{j}^{\mathsf{K}} \left( X^{l\mathsf{K}} \right) \right) + \mathsf{K}(c^{*}_{j} + \delta)
\end{eqnarray*}

where $(a)$ follows from~(\ref{eq:error function}).

Since $ \lim_{n \rightarrow \infty} \frac{\Delta^{n\mathsf{K}}
V_{j}(x)}{n\mathsf{K}} = c_{j}^{*}$, we have that $\frac{1}{n} \sum_{l=0}^{n-1}
\mathsf{E}_{x} \left( g_{j}^{\mathsf{K}} \left( X^{l\mathsf{K}} \right) \right)
=\mathsf{K} \delta$. Now

\begin{eqnarray*}
\liminf_{n \rightarrow \infty} \frac{1}{n} p_{x}
\left( \max_{j} g_{j}^{\mathsf{K}}(x) < \mathsf{K}\left( c^{*} + \delta +
\epsilon \right) \right) & \stackrel{(b)}{\geq} &
1 - \limsup_{n \rightarrow \infty} \frac{1}{n} \sum_{l=0}^{n-1} \sum_{j=1}^{J} \frac{\mathsf{E}_{x} \left( g_{j}^{\mathsf{K}} \left( X^{l\mathsf{K}} \right) \right)}{\mathsf{K}\left( c^{*} + \delta + \epsilon \right)}  \\
& \geq & 1 - \frac{1}{\mathsf{K}\left( c^{*} + \delta + \epsilon \right)} \times \\
&& \sum_{j=1}^{J}  \limsup_{n \rightarrow \infty} \frac{1}{n} \sum_{l=0}^{n-1} \mathsf{E}_{x} \left( g_{j}^{\mathsf{K}} \left( X^{l\mathsf{K}} \right) \right)  \\
 &=& 1 - \frac{J\mathsf{K} \delta}{\mathsf{K}\left( c^{*} + \delta + \epsilon \right)} \\
 &=& 1 - \frac{J\delta}{\left( c^{*} + \delta + \epsilon \right)}, \\
\end{eqnarray*}
where $(b)$ follows from Lemma 1 in \cite{Ros-JRN-JAP-1980}.

We note that there exists a $\delta_{0} > 0$ such that $1 - \frac{J\delta}{\left( c^{*} + \delta_{0} + \epsilon \right)} > 0$. Define the set
\begin{eqnarray*}
\mathcal{A}_{0} &=& \left\{ x \in \mathcal{X}: \max_{j} g_{j}^{\mathsf{K}}(x) < \mathsf{K}\left( c^{*} + \delta_{0} + \epsilon \right) \right\}
\end{eqnarray*}
We can observe that $\mathcal{A}_{0} \subseteq \cap_{j} \mathcal{A}_{j,k}$ is a finite set. Hence it follows that for the finite set $\mathcal{A}_{0}$,

\begin{eqnarray*}
\liminf_{n \rightarrow \infty} \frac{1}{n} \sum_{l=0}^{n-1} p_{x} \left( X^{l\mathsf{K}} \in \mathcal{A}_{0} \right) & > & 0
\end{eqnarray*}

Since the chain is assumed to be irreducible and aperiodic, it follows that the
Markov chain is positive recurrent.  \end{proof}

\section{Multiclass  Network Model} \label{sec:model}
Consider an open queueing network consisting of $J$ work conserving single-servers, each with an infinite capacity queue. Let $\alpha=1,2,\ldots $ denote the class of an arriving customer and customers of class $\alpha$ arrive at server $j, j=1,2,\ldots, J$ as a Poisson process with rate $\lambda_{\alpha,j}$ and a customer of class $\alpha$ at node $j$ requires a service time that  is exponentially distributed with rate $\mu_{\alpha,j}$. After completing service a customer of class $\alpha$ is routed to server $k$ with probability $r_{j,k}$ that is the \emph{same} for all classes.
It is assumed that the scheduling policy at any server is a stationary, \emph{work-conserving} policy.

One way of constructing this network is to consider a model in which customers from outside to the network arrive  in a Poisson process of rate $\lambda$. On arrival, the customer joins the queue attached to server $j$ as a class $\alpha$ customer with the probability $q_{\alpha, j}$ such that $\sum_{\alpha, j} q_{\alpha, j}=1$. Thus $\lambda_{\alpha,j}= \lambda q_{\alpha,j}$.
 Once assigned, a customer will retain its class  till it exits the network. Let $R=[r_{j, k}]$ denote the routing probability matrix and by construction $R^T$ is \emph{invertible} since by assumption the network is open.


Let $\left\{ \Lambda_{\alpha, j} \right\}$, where $\Lambda_{\alpha, j}$ is the equilibrium rate at which customers of the class $\alpha$ arrive to the $j$th server, be the solution of the following set of equations
\begin{eqnarray}
\label{eq:traffic equation}
\Lambda_{\alpha, j} &=& \lambda q_{\alpha, j} + \sum_{k=1}^{J} \Lambda_{\alpha, k} \; r_{k, j}
\end{eqnarray}
written for each class $\alpha$ and for each server $j$.

Let us define $\rho_{\alpha, j}$ as the load on  server $j$ from  customers of the class $\alpha$, then
\begin{eqnarray*}
\rho_{\alpha, j} &=&  \frac{\Lambda_{\alpha, j}}{\mu_{\alpha, j}}
\end{eqnarray*}
Let  $\Gamma_{\alpha, j}^{k}$ denote the expected number of times a customer of the class $\alpha$, who is presently in the queue $j$, visits the server $k$ before it departs the network. Then $\Gamma_{\alpha, j}^{k}$ satisfies the following equations:
\begin{equation}
\label{eq:conservation equations}
\left.
\begin{aligned}
\Gamma_{\alpha, j}^{k} &= \delta_{j, k} + \sum_{l=1}^{J} r_{j, l} \Gamma_{\alpha, l}^{k} \\
\sum_{\alpha} \sum_{k} \lambda q_{\alpha, k} \Gamma_{\alpha, k}^{j} &= \sum_{\alpha} \lambda_{\alpha, j}
\end{aligned}
\right\}
\end{equation}
where $\delta_{j, k}$ is the Kronecker Delta function. Let $\mathcal{X}$ denote the countable state space for the Markov chain $\{X_{n}, n \geq 1\}$, modeling the evolution of the numbers of customers of different classes in various queues. We assume that the state definition is detailed enough that in every state $x$ we can obtain the class and positional information of every customer in every queue. Let $\{p_{x, y}, x,y \in \mathcal{X}\}$ denote the transition probabilities of the Markov chain.

In the following we show that for a given multiclass network one can find a multiclass network but with \emph{single customer service rate} (i.e., $\mu_{\alpha, j} = \mu$) such that stability of the \emph{single customer service rate} network implies stability of the actual multiclass network. We will see that this reduction of multi service rates to single service rate coupled with the class-independent routing assumption helps us easily establish some monotonicity properties (Lemma \ref{lemma:verification of the UUB condition}) in the single service rate network.

Let $\mathcal{S}$ represent a network in which the Poisson arrival rate $\lambda$ and the service rates $\{\mu_{\alpha, j}\}$ are such that $\sum_{\alpha} \rho_{\alpha, j} < 1$ for each server $j$.
We now construct a network $\mathcal{S}^{\prime}$
which will be identical to the actual network $\mathcal{S}$ except for the Poisson  arrival rate and average service times of the customers. Specifically, we will choose a $\lambda^{\prime}$ and $\mu$ such that
$\rho_{\alpha, j} < \frac{\Lambda_{\alpha, j}^{\prime}}{\mu} < \rho_{\alpha, j}+\eta_{\alpha, j} < 1$ and $\sum_{\alpha} \rho^{\prime}_{\alpha, j} = \frac{1}{\mu} \sum_{\alpha} \Lambda_{\alpha, j}^{\prime} < 1$ hold for all classes $\alpha$ and all servers $j$  where $\eta_{\alpha, j}$ are arbitrarily small positive numbers. We note here that $\Lambda_{\alpha, j}^{\prime}$ solve the traffic equation (\ref{eq:traffic equation}) with $\lambda^{\prime}$ in place of $\lambda$.
Equivalently
\[
\max_{\alpha, j} \frac{\Lambda_{\alpha, j}^{\prime}}{\rho_{\alpha, j}+\eta_{\alpha, j}} < \mu < \min_{\alpha, j} \mu_{\alpha, j}, \quad \mbox{and} \quad \mu > \max_{j} \sum_{\alpha} \Lambda_{\alpha, j}^{\prime}
\]
should hold. This becomes feasible when $\lambda^{\prime}$ is chosen sufficiently small. Thus, we have constructed a network $\mathcal{S}^{\prime}$ in which customers arrive at a lower rate and spend longer times on an average in each of the servers they visit. Furthermore, the load due to class $\alpha$ at the server $j$ in $\mathcal{S}^{\prime}$ is larger than the corresponding in the actual network $\mathcal{S}$. We remark here that even though customers in the network $\mathcal{S}^{\prime}$ have the same service rate $\mu$ they are still identifiable by their class.

Following the uniformization technique \cite{Lip-OR-1975}, we will now obtain the discrete-time Markov chain queueing models
$\{X_{n}, n \geq 0\}$ and $\{Y_{n}, n \geq 0\}$ of the multiclass networks $\mathcal{S}$ and $\mathcal{S}^{\prime}$, respectively.
For the purposes of uniformization, let us choose two numbers $Q_{1}$ and $Q_{2}$ such that $Q_{1} > \lambda + J \max_{\alpha, j} \mu_{\alpha, j}$ and $Q_{2} > \lambda^{\prime} + J\mu$ and $\frac{\lambda}{Q_{1}} = \frac{\lambda^{\prime}}{Q_{2}}$ where $Q_{1}$ and $Q_{2}$ are the respective Poisson rates that achieve uniformization for the networks $\mathcal{S}$ and $\mathcal{S}^{\prime}$, respectively. Note that $Q_{1}$ and $Q_{2}$ are so chosen that at every time epoch the probability that a customer arrives is \emph{same} in both the chains.

With this construction we now have two $\{X_{n}\}$ and $\{Y_{n}\}$  that are the same except for the probability for  service completion at a given time epoch: \emph{all servers in the network $\mathcal{S}^{\prime}$ are slower than the corresponding servers in the network $\mathcal{S}$}.

Next we demonstrate a construction in which the Markov chains
$\{X_{n}, n \geq 0\}$ and $\{Y_{n}, n \geq 0\}$ will be defined on a common probability space and then we show that the total number of customers in the network $\mathcal{S}^{\prime}$ will always be at least as large as the total number of customers in the network $\mathcal{S}$ under the condition that both networks are initialized in the \emph{same state}. Let $\{A_{n}, n \geq 1\}$ be a sequence of i.i.d. random variables with the distribution $p(A_{n}=1) = \frac{\lambda}{T_{1}} = \frac{\lambda^{\prime}}{T_{2}} = 1- p(A_{n}=0)$ where  $A_{n}=1$ denotes an arrival at the $n$th epoch in \emph{both} the networks.
For $1 \leq j \leq J$, let $\{B_{j, k}, k \geq 1\}$ be a sequence of i.i.d. random variables with the distribution $p(B_{j, k} = l) = r_{j, l}$ where $B_{j, k}$ denotes the server to be chosen by the $k$th departure from the $j$th server in \emph{both} the networks. Also, by suitable construction we can generate departure events so that a when the $j$th server is non-empty in $\mathcal{S}$ as well as $\mathcal{S}^{\prime}$ then a departure from the $j$th server
in $\mathcal{S}^{\prime}$ implies a departure from the $j$th server in $\mathcal{S}$. An implication of this construction is that the departure epoch of $k$th customer from the $j$th server in $\mathcal{S}^{\prime}$ will be \emph{no sooner} than the corresponding departure in the network $\mathcal{S}$. Thus the total number of customers in the network $\mathcal{S}^{\prime}$ will always be at least as large as the total number of customers in the network $\mathcal{S}$. Thus stability of the network $\mathcal{S}^{\prime}$ implies stability of the actual network $\mathcal{S}$.

Following the arguments above, from now on we concentrate exclusively on the network with a single customer service rate.

\begin{theorem}
\label{th:my theorem}
The condition $\sum_{\alpha} \rho_{\alpha, j} < 1$ for $1 \leq j \leq J$ is a sufficient condition for positive recurrence of the Markov chain $\{X_{n}, n \geq 1\}$. \Q
\end{theorem}
Before we prove Theorem \ref{th:my theorem}, we need to establish the following Lemma.

\begin{lemma}
\label{lemma:verification of the UUB condition}
Let $x \in \mathcal{X}$ be a non-zero state and $p_{x, \mathcal{X}_{j}}^{n} = \sum_{y \in \mathcal{X}_{j}} p_{x, y}^{n}$. Then $p_{x, \mathcal{X}_{j}}^{n} \leq p_{0, \mathcal{X}_{j}}^{n}$ under class-independent routing. \Q
\end{lemma}

\begin{proof}
Let the notation $X^{x}_{n}$, $n \geq 1$, represent the Markov chain whose state at time $n=0$ is $x$. Let $y_{j}^{n}(x) = \sum_{\alpha} x_{\alpha, j}^{n}$ denote the total number of customers in the $j$th queue after $n$ transitions when the state at $n=0$ is $x$. Let now $x$ be a non-zero state and 0, the zero state. We need to show that $y_{j}^{n}(x)$ is \emph{stochastically larger} than $y_{j}^{n}(0)$; in notation $y_{j}^{n}(x) \geq_{\mbox{s.t.}} y_{j}^{n}(0)$. But this implies that $p_{x, \mathcal{X}_{j}}^{n} \leq
p_{0, \mathcal{X}_{j}}^{n}$ for $n \geq 1$.
Following Strassen's theorem \cite{BacBre-BOOK} it is enough to show the same on a \emph{common probability space} formed by the customer arrival process, random variables modeling service times of customers of different classes at different servers, and also the random variables that model routing decisions of customers at each of the $J$ servers.

Because of the assumption that customers routing decisions are independent of their class, it is not hard to see that $y_{j}^{n}(x) \geq_{\mbox{s.t.}} y_{j}^{n}(0)$ and hence
 $p_{x, \mathcal{X}_{j}}^{n} \leq p_{0, \mathcal{X}_{j}}^{n}$, $n \geq 1$ given that $x$ is a non-zero state. That is, the total number of customers in each queue in a non-zero state always dominate the corresponding quantities in a smaller state.
\end{proof}

{\bf Proof of Theorem \ref{th:my theorem}:}
By letting $x_{\alpha, k}$ denote the number of customers of the class $\alpha$ at the server $k$, we consider the Lyapunov functions $V_{j}(x)$, one for each server $j$, that were originally proposed in \cite{Ros-JRN-JAP-1980}.
\begin{eqnarray*}
V_{j}(x) &=& \sum_{\alpha} \sum_{k=1}^{J} x_{\alpha, k} \Gamma_{\alpha, k}^{j}
\end{eqnarray*}
The Lyapunov function $V_{j}(x)$ can be interpreted as the \emph{virtual work load} on the server $j$ in state $x$, whereas $x_{j}= \sum_{\alpha} x_{\alpha, j}$, the total number of customers who are presently in the queue $j$, denotes the \emph{physical work load} on the server $j$. We now define the sets $\mathcal{X}_{j} = \{x \in \mathcal{X}: x_{j} = 0\}$ and then $\mathcal{X}_{j}^{c} = \{x \in \mathcal{X}: x_{j} \geq 1\}$ so that $\left\{ \mathcal{X}_{j}, \mathcal{X}_{j}^{c} \right\}$ is a partition of the state space $\mathcal{X}$.
Define $Q = \lambda + J\mu$. Our next task is establish the drift
\begin{eqnarray*}
\Delta V_{j}(x) &=& \left\{
\begin{array}{ll}
\frac{1}{Q} \left( \sum_{\alpha} \Lambda_{\alpha, j} - \mu \right) & \mbox{if $x \in \mathcal{X}_{j}^{c}$} \\
\frac{1}{Q} \sum_{\alpha} \Lambda_{\alpha, j} & \mbox{if $x \in \mathcal{X}_{j}$}
\end{array}
\right.
\end{eqnarray*}
We identify the following events in the system that result in a state transition:
\begin{itemize}
\item[E1]   An exogenous customer arrival of class $\alpha$ to the server $j$ with probability $\frac{\lambda_{\alpha, j}}{Q}$,
\item[E2]  A departure from the $j$th server that moves to the $k$th server with probability $\frac{\mu r_{j, k}}{Q}$,
\item[E3] A departure from the $j$th server that exits the network with probability $\frac{\mu (1-\sum_{k=1}^{J} r_{j, k})}{Q}$ and,
\item[E4] A transition that takes the chain back to the same state with probability $1 - \frac{\lambda}{Q} - \sum_{j:x_{j} \neq 0} \frac{\mu}{Q}$.
\end{itemize}

Then we can derive the drift $\Delta V_{j}(x)$ as

\begin{eqnarray*}
Q \Delta V_{j}(x) &=& \underbrace{\sum_{\alpha} \sum_{k=1}^{J} \Gamma_{\alpha, k}^{j} \lambda q_{\alpha, k}}_{E1} +
\underbrace{\sum_{\alpha} \sum_{m, n} \left( \Gamma_{\alpha, n}^{j} - \Gamma_{\alpha, m}^{j} \right) \mu r_{m, n}}_{E2} +
\underbrace{\sum_{\alpha} \sum_{m} - \Gamma_{\alpha, m}^{j} \mu [1 -\sum_{k=1}^{J} r_{m, k}]}_{E3} \\
&\stackrel{(*)}{=}& \sum_{\alpha} \Lambda_{\alpha, j} - \mu \mathbb{I}\{x_{j} > 0\}
\end{eqnarray*}
where $(*)$ follows from the equations (\ref{eq:conservation equations}).

To verify the UUB condition (Assumption \ref{assumption 3} of Theorem \ref{th:My generalization of Rosberg's theorem}), we first note that
\begin{eqnarray*}
Q \frac{\Delta V_{j}^{n}(x)}{n} &=& \sum_{y \in \mathcal{X}_{j}} Q \Delta V_{j}(y)  \frac{1}{n} \sum_{k=0}^{n-1} p_{x, y}^{k} + \sum_{y \in \mathcal{X}_{j}^{c}} Q \Delta V_{j}(y) \frac{1}{n} \sum_{k=0}^{n-1} p_{x, y}^{k} \\
&=& \left( \sum_{\alpha} \Lambda_{\alpha, j} \right) \frac{1}{n}  \sum_{k=0}^{n-1} p_{x, \mathcal{X}_{j}}^{k} + \left( \sum_{\alpha} \Lambda_{\alpha, j} - \mu  \right) \frac{1}{n} \sum_{k=0}^{n-1} p_{x, \mathcal{X}_{j}^{c}}^{k}  \\
&=&  \sum_{\alpha} \Lambda_{\alpha, j} - \mu \frac{1}{n} \sum_{k=0}^{n-1} p_{x, \mathcal{X}_{j}^{c}}^{k}
\end{eqnarray*}

From Lemma \ref{lemma:verification of the UUB condition} this implies that $\frac{1}{n} \Delta V_{j}^{n}(x) \leq \frac{1}{n} \Delta V_{j}^{n}(0)$, thus proving the UUB condition.
\Q

\bibliographystyle{unsrt}
\bibliographystyle{apt}
\bibliography{info}

\begin{thebibliography}{10}

\bibitem{Ros-JRN-JAP-1980}
Z.~Rosberg.
\newblock ``{A} {P}ositive {R}ecurrence {C}riterion {A}ssociated with
  {M}ultidimensional {Q}ueueing {P}rocesses''.
\newblock {\em Journal of Applied Probability}, 17(3):790--801, 1980.

\bibitem{LuKum-JRN-ACTRAN}
S.~H. Lu and P.~R Kumar.
\newblock Distributed scheduling based on due dates and buffer priorities.
\newblock {\em IEEE Trans. Automat. Control}, 36:1406--1416, 1991.

\bibitem{KumSei-JRN-ACTRAN}
P.~R. Kumar and T.~I Seidman.
\newblock Dynamic instabilities and stabilization methods in distributed
  real-time scheduling of manufacturing systems.
\newblock {\em IEEE Trans. Automat. Control}, 35:289--298, 1990.

\bibitem{KumMey-JRN-JSAC-1995}
P.~R Kumar and S~Meyn.
\newblock Stability of {Q}ueueing {N}etworks and {S}cheduling {P}olicies.
\newblock {\em IEEE Trans. Automat. Control}, 40(2):251--260, February 1995.

\bibitem{RybSto-JRN-PRI}
A.~N. Rybko and A.~L Stolyar.
\newblock On the ergodicity of random processes that describe the functioning
  of open queueing networks.
\newblock {\em Problemy Peredachi Informatsii}, 28(3-26), 1992.

\bibitem{Bra-AAP-1994}
M~Bramson.
\newblock Instability of {FIFO} queueing networks.
\newblock {\em Ann. Appl. Probab}, 4:414--431, 1994.

\bibitem{Sei-JRN-ACTRAN}
T.~I Seidman.
\newblock First come, first served" can be unstable!
\newblock {\em IEEE Trans. Automat. Control}, 39, 1994.

\bibitem{Dai-AAP-1995}
J.~G. Dai.
\newblock On positive harris recurrence of multiclass queueing networks: A
  unified approach via fluid limit models.
\newblock {\em Ann. Appl. Probab.}, 5:49--77, 1995.

\bibitem{Sto-MPRF-1995}
A.~L Stolyar.
\newblock On the stability of multiclass queueing networks: A relaxed
  sufficient condition via limiting fluid processes.
\newblock {\em Markov Process. Related Fields}, 1:491--512, 1995.

\bibitem{Bra-QUESTA-1996}
M~Bramson.
\newblock Convergence to equilibria for fluid models of fifo queueing networks.
\newblock {\em Queueing Systems Theory Appl.}, 22:5--45, 1996.

\bibitem{FosRyb-MPRF-1996}
S.~Foss and A.~Rybko.
\newblock Stability of {M}ulticlass {J}ackson-{T}ype {N}etworks.
\newblock {\em Markov Processes and Related Fields}, 2(3):461--487, 1996.

\bibitem{Kelly-BOOK}
F.~P. Kelly.
\newblock {\em Reversibility and {S}tochastic {N}etworks}.
\newblock John Wiley \& Sons Ltd, 1979.

\bibitem{Foster_1953}
F.~G. Foster.
\newblock ``{O}n the stochastic matrices associated with certain queueing
  processes''.
\newblock {\em Ann. Math. Statist}, 24:355--360, 1953.

\bibitem{Lip-OR-1975}
S.A Lippman.
\newblock Applying a new devise in optimization of exponential queueing
  systems.
\newblock {\em Oper. Res.}, 23:687--710, 1975.

\bibitem{BacBre-BOOK}
F.~Baccelli and P.~Bremaud.
\newblock {\em {E}lements of {Q}ueueing {T}heory}, volume~26 of {\em
  Applications of Mathematics}.
\newblock Springer, second edition.

\end{thebibliography}

\end{document}